\documentclass[a4paper,UKenglish]{lipics-v2018}

\usepackage{times}
\usepackage{amssymb}
\usepackage{caption}
\usepackage{smallsec}

\usepackage{setspace}

\usepackage{tikz}

\usetikzlibrary{arrows,shapes,snakes,automata,backgrounds,positioning}

\tikzstyle{player1}=[draw,rounded rectangle, minimum size=5mm]
\tikzstyle{player2}=[draw,rectangle,minimum size=5mm]
\tikzstyle{widget}=[draw=red,rectangle, rounded rectangle=10pt,dashed,minimum size=6mm,fill=yellow]
\tikzset{every loop/.style={looseness=7}}

\usetikzlibrary{external} 
\newcommand{\zug}[1]{\langle #1  \rangle}
\newcommand{\stam}[1]{}
\newcommand{\realpos}{\mbox{I$\!$R}_{\ge 0}}

\newcommand{\Nat}{\mathbb{N}}

\newcommand{\Q}{\mathbb{Q}}
\newcommand{\A}{{\cal A}}
\newcommand{\N}{{\cal N}}
\newcommand{\T}{{\cal T}}
\newcommand{\visits}{{\it visits}}
\newcommand{\load}{{\it load}}

\newcommand{\nextT}{{\it next}_T}
\newcommand{\F}{{\cal F}}
\newcommand{\G}{{\cal G}}

\renewcommand{\P}{{\cal P}}
\newcommand{\R}{{\cal R}}
\newcommand{\Prof}{{\it profiles}}
\newcommand{\K}{{\cal K}}

\usepackage{xspace}
\newcommand{\PO}{Player~$1$\xspace}
\newcommand{\PT}{Player~$2$\xspace}

\newcommand{\set}[1]{\{ #1 \} }

\nolinenumbers

\title{Timed Network Games with Clocks}

\author{Guy Avni}{IST Austria, Klosterneuburg, Austria}{guy.avni@ist.ac.at}{}{Supported by the Austrian Science Fund (FWF) under grants S11402-N23 (RiSE/SHiNE), Z211-N23 (Wittgenstein Award), and M2369-N33 (Meitner fellowship).}
\author{Shibashis Guha}{Universit\'{e} Libre de Bruxelles, Brussels, Belgium}{shibashis.guha@ulb.ac.be}{}{Supported partially by the ARC project ``Non-Zero Sum Game Graphs: Applications to Reactive Synthesis and Beyond'' (F\'{e}d\'{e}ration Wallonie-Bruxelles).}
\author{Orna Kupferman}{Hebrew University, Jerusalem, Israel}{orna@cs.huji.ac.il}{}{Supported by the European Research Council  (FP7/2007-2013) / ERC grant agreement no 278410.}

\authorrunning{G. Avni, S. Guha and O. Kupferman} 

\Copyright{Guy Avni, Shibashis Guha and Orna Kupferman}

\subjclass{Theory of computation -- Theory and algorithms for application domains -- Algorithmic game theory and mechanism design -- Algorithmic game theory,
Theory of computation -- Theory and algorithms for application domains -- Algorithmic game theory and mechanism design -- Network formation,
Theory of computation -- Models of computation -- Timed and hybrid models
}

\begin{document}

\maketitle

\begin{abstract}
Network games are widely used as a model for selfish resource-allocation problems. In the classical model, each player selects a path connecting her source and target vertices. The cost of traversing an edge depends on the {\em load}; namely, number of players that traverse it. Thus, it abstracts the fact that different users may use a resource at different times and for different durations, which plays an important role in determining the costs of the users in reality. For example, when transmitting packets in a communication network, routing traffic in a road network, or processing a task in a production system, actual sharing and congestion of resources crucially depends on time.

In \cite{AGK17}, we introduced {\em timed network games}, which add a time component to network games. Each vertex $v$ in the network is associated with a cost function, mapping the load on $v$ to the price that a player pays for staying in $v$ for one time unit with this load. 
Each edge in the network is guarded by the time intervals in which it can be traversed, which forces the players to spend time in the vertices. In this work we significantly extend the way time can be referred to in timed network games. In the model we study, the network is equipped with {\em clocks}, and, as in timed automata, edges are guarded by constraints on the values of the clocks, and their traversal may involve a reset of some clocks. 
We argue that the stronger model captures many realistic networks. The addition of clocks breaks the techniques we developed in \cite{AGK17} and we develop new techniques in order to show that positive results on classic network games carry over to the stronger timed setting.
\end{abstract}
\stam{
The model in \cite{AGK17} can be viewed as a special case of the model studied here, with a single clock that is never reset. We argue that the stronger model captures many realistic networks, and we solve the challenging technical problems it involves. In particular, as clocks may be reset, there is no bound on the duration of strategies, and a partitioning of the infinite space of clock values to a finite one is not always possible. We are still able to obtain positive results about the stability of timed network games and present optimal algorithms to classical problems about network games in the stronger timed setting.
}

\section{Introduction}
\label{intro}
Network games (NGs, for short) \cite{AD+08,Ros73,RT02} constitute a well studied model of non-cooperative games. The game is played among selfish players on a network, which is a directed graph. Each player has a source and a target vertex, and a strategy is a choice of a path that connects these two vertices. The cost a player pays for an edge depends on the {\em load} on it, namely the number of players that use the edge, and the total cost is the sum of costs of the edges she uses. In {\em cost-sharing} games, load has a positive effect on cost: each edge has a cost and the players that use it split the cost among them. Then, in {\em congestion} games\footnote{The name {\em congestion games} is sometimes used to refer to games with general latency functions. We find it more appropriate to use it to refer to games with non-decreasing functions.}, load has a negative effect on cost: each edge has a non-decreasing {\em latency function\/}
that maps the load on the edge to its cost.

One limitation of NGs is that the cost of using a resource abstracts the fact that different users may use the resource at different times and for different durations. This is a real limitation, as time plays an important role in many real-life settings.
For example, in a road or a communication system, congestion only affects cars or messages that use a road or a channel simultaneously.
We are interested in settings in which congestion affects the quality of service (QoS) or the way a price is shared 
by entities using a resource at the \emph{same time} (rather than affecting the travel time).
For example, discomfort increases in a crowded train (in congestion games) or price is shared by the passengers in a taxi (in cost-sharing games).
\stam{
Similarly, in mobile networks, the call quality depends on the number of subscribers using the network simultaneously. As a third example, when processing a task in a production system, jobs move from one station to another. The way the cost of running the stations is shared by the jobs that use it depends on the time spent in the stations and on the synchronization among the jobs.
}

The need to address temporal behaviors has attracted a lot of research in theoretical computer science.  Formalisms like temporal logic \cite{Pnu81} enable the specification of the temporal ordering of events. Its refinement to formalisms like real-time temporal logic \cite{AH90}, interval temporal logic \cite{MM83}, and timed automata (TAs, for short) \cite{AD94}  enables the specification of {\em real-time\/} behaviors. Extensions of TAs include {\em priced timed automata} (PTAs, for short) that assign costs to real-time behaviors. Thus, PTAs are suitable for reasoning about quality of real-time systems. They lack, however, the capability to reason about multi-agent systems in which the players' choices affect the incurred costs.
\stam{
Early work on timed aspects of networks focus on  {\em flow networks\/} \cite{FF62}, and spans from theoretical work on queues (c.f., \cite{Vic69,Yag71}) to nowadays practical research on traffic engineering in software defined networks \cite{AKL13b}.
}

We study {\em timed network games} (TNGs, for short) --
a new model that adds a time component to NGs. A TNG is played on a {\em timed-network} in which edges are labeled by {\em guards} that specify time restrictions on when the edge can be traversed. Similar to NGs, each player has a source and target vertex, but a strategy is now a {\em timed path} that specifies, in addition to which vertices are traversed, the amount of time that is spent in each vertex. Players pay for staying in vertices, and the cost of staying in a vertex $v$ in a time interval $I \subseteq \realpos$
is affected by the load in $v$ during $I$. In \cite{AGK17}, we studied a class of TNGs that offered a first extension of NGs to a timed variant in which the reference to time is restricted: the guards on the edges refer only to {\em global} time, i.e., the time that has elapsed since the beginning of the game. In the model in \cite{AGK17}, it is impossible to refer to the duration of certain events that occur during the game, for example, it is not possible to express constraints that require staying exactly one time unit in a vertex. Accordingly, we refer to that class as {\em global TNGs} (GTNGs, for short).

In this work, we significantly extend the way time can be referred to in TNGs. 
We do this by adding {\em clocks\/} that may be reset along the edges,
and by allowing the guards on the edges to refer to the values of all clocks. GTNGs can be viewed as a fragment in which there is only a single clock that is never reset. We demonstrate our model in the following example.

\begin{example}
\label{first example}
Consider a setting in which messages are sent through a network of routers. Messages are owned by selfish 
agents who try to avoid congested routes, where there is a greater chance of 
loss or corruption. The owners of the messages decide how much time they spend in each router. Using TNGs, we can model constraints on these times, as well as constraints on global events, in particular, arrival time. Note that in some applications, c.f., advertising or security, messages need to patrol the network with a lower bound on their arrival time.

Consider the TNG appearing in Figure~\ref{fig:ex1}. The vertices in the TNG model the routers.
There are two players that model two agents, each sending a message. 
The source of both messages is $s$ and the targets are $u_1$ and $u_2$, for messages~$1$ and~$2$, respectively.  
The latency functions are described in the vertices, as a function of the load $m$; e.g., the latency function in $v_2$ is $\ell_{v_2}(m) = 3m$. Thus, when a single message stays in $v_2$ the cost for each time unit is $3$, and when the two messages visit $v_2$ simultaneously, the cost for each of them is $6$ per unit time. The network has two clocks, $x$ and $y$. 
Clock $x$ is reset in each transition and thus is used to impose restrictions on the time that can be spent in each router: since all transitions can be taken when $1 \leq x \leq 2$, a message stays between $1$ and $2$ time units in a router.  
Clock $y$ is never reset, thus it keeps track of the global time. 
The guards on clock $y$ guarantee that message $1$ reaches its destination by time $4$ but not before time $3$ and message $2$ reaches its destination by time $5$ but not before time $4$.

Suppose the first agent chooses the timed path $(s,2), (v_1,1), u_1$, thus message $1$ stays in $s$ for two time units and in $v_1$ for one time unit before reaching its destination $u_1$. Suppose the second agent chooses the path $(s,2), (v_1, 2), (v_2,1), u_2$. Note that crossing an edge is instantaneous. Since both messages stay in the same vertices during the intervals $I_1 = [0,2]$ and $I_2 = [2,3]$, the load in the corresponding vertices is $2$. 
During interval $I_1$, each of the agents pays $|I_1| \cdot \ell_s(2) = 2 \cdot 4$ and during $I_2$, each pays $|I_2| \cdot \ell_{v_1}(2) = 1 \cdot 2$. 
Message $2$ stays in $v_1$ alone during the interval $[3,4]$ and in $v_2$ during the interval $[4,5]$, for which it pays $1$ and $3$, respectively.
The total costs are thus $10$ and $14$.\hfill\qed
\begin{figure}[ht]
\center
\includegraphics[height=1.7cm]{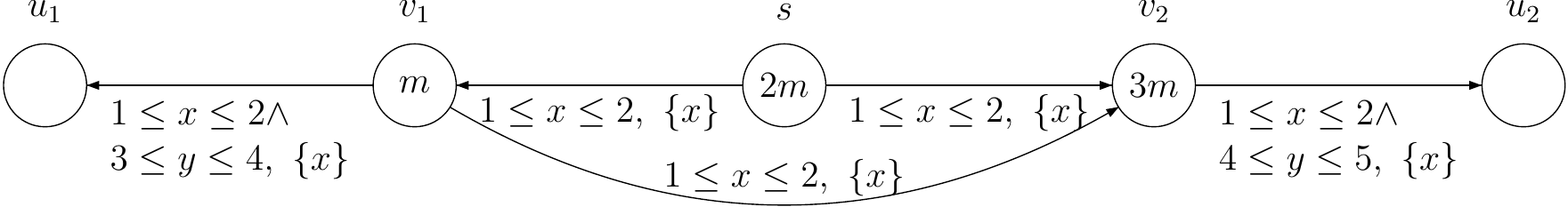}
\caption{A congestion TNG.}
\label{fig:ex1}
\end{figure}
\end{example}

\vspace{-.4cm}
Before we elaborate on our contribution, let us survey relevant works, namely, extensions of NGs with temporal aspects and extensions of timed-automata to games. 
Extensions of NGs that involve reasoning about time mostly study a cost model in which the players try to minimize the time of arrival at their destinations (c.f., \cite{HMRT11,KP12,KS11,PPT09}), 
where, for example, congestion affects the duration of crossing an edge. These works are different from ours since we consider a QoS cost model. 
An exception is \cite{HMRT11}, which studies the QoS costs. A key difference in the models is that there, time is discrete and the players have finitely many strategies.
Thus, reductions to classical resource allocation games is straightforward while for TNGs it is not possible, as we elaborate below.
Games on timed automata were first studied in \cite{AM99} in which an algorithm to solve timed games with timed reachability objective was given. The work was later generalized and improved \cite{ABM04,BCFL04,HIM13,BGKMMT14}.
Average timed games, games with parity objectives, mean-payoff games and energy games have also been studied in the context of timed automata \cite{AFHMS03,JT08,CHP11,BCR14,GJKT16}. All the timed games above are two-player zero-sum ongoing games. Prices are fixed and there is no notion of load. Also, the questions studied on these games concern their decidability, namely finding winners and strategies for them.
TNGs are not zero-sum games, so winning strategies do not exist. Instead, the problems we study here concern rationality and stability.

The first question that arises in the context of non-zero-sum games is the existence of {\em stable outcomes}. In the context of NGs, the most prominent stability concept is that of a (pure) {\em Nash equilibrium\/} (NE, for short) \cite{Nas50} -- a profile such that no player can decrease her cost by unilaterally deviating from her current strategy.\footnote{Throughout this paper, we consider {\em pure\/}
strategies, as is the case for the vast literature on NGs. }
Decentralized decision-making may lead to
solutions that are sub-optimal for the society as a whole. The standard measures to quantify the inefficiency incurred due to selfish behavior is the {\em price of  stability} (PoS) \cite{AD+08} and the {\em price of anarchy} (PoA) \cite{KP09b}. In both measures we compare against the {\em social optimum} (SO, for short), namely a profile that minimizes the sum of costs of all players. The PoS (PoA, respectively)  is the best-case (worst-case) inefficiency of an NE;
that is, the ratio between the cost of a best (worst) NE and the SO.
\stam{
In Example \ref{first example}, profile $P$ is an SO, and is also a (best) NE,
while profile $P'$ is a worst NE. Note that there can be uncountably many NEs in the TNG in Example \ref{first example}.
Indeed, for all $t \in [0,1]$, the profile $P'_t$ with the strategies
$(s,2),(a,4),(w,3+t),(a,4-t),u_1$ and
$(s,2),(a,4),(w,3+t),(a,4),(w,1-t),u_2$
is an NE with costs $8/2 \cdot 2 + 4/2 \cdot 4 + 6/2 \cdot (3+t) + 4/2 \cdot (4-t) = 33+t$ and
$8/2 \cdot 2 + 4/2 \cdot 4 + 6/2 \cdot (3+t) + (4/2 \cdot (4-t) + 4 \cdot t) + 6 \cdot (1-t) = 39-t$, respectively.
}

The picture of stability and equilibrium inefficiency for standard NGs is well understood. Every NG has an NE, and in fact these games are {\em potential games} \cite{Ros73}, 
which have the following stronger property: a {\em best response sequence\/} is a sequence of profiles $P_1,P_2,\ldots$ such that, for $i \geq 1$, the profile $P_{i+1}$ is obtained from $P_i$ by letting some player deviate and decrease 
her
personal cost. In finite potential games, every best-response sequence converges to an NE. 
For $k$-player cost-sharing NGs, the PoS and PoA are $\log k$ and $k$, respectively \cite{AD+08}. For congestion games with affine cost functions, $\mbox{PoS} \approx 1.577$ \cite{CK05,ADGMS11} and $\mbox{PoA}=\frac{5}{2}$ \cite{CK05b}. 

In \cite{AGK17}, we showed that these positive results carry over to GTNGs. A key technical feature of GTNGs 
is that since guards refer to global time, it is easy to find an upper bound $T$ on the time by which all players reach their destinations. Proving existence of NE follows from a reduction to NGs, 
using a zone-like structure \cite{ACHHHNOSY95,BY03}. 
The introduction of clocks with resets breaks the direct reduction to NGs and questions the existence of a  bound by which the players arrive at their destinations.
Even with an upper bound on time, a reduction from TNGs to NGs is not likely.
Consider the following example. From $s_1$, the earliest absolute time at which vertex $v_2$ is reached is $2$
following the path $\zug{(s_1,0)(v_1,0),(v_1,2),(v_2,2)}$, and the value of clock $x$ at $v_2$ is $0$.
On the other hand, when $v_2$ is reached from $s_2$ following the path $\zug{(s_2,0)(v_2,0),(v_2,2)}$,
then at absolute time $2$, the value of clock $x$ at $v_2$ is $2$ and the transition to vertex $u$
is thus enabled.
This leads to a spurious path $\zug{(s_1,0)(v_1,0),(v_1,2),(v_2,2),(u,2)}$ which 
does not correspond to a valid path in th TNG.

\begin{figure}[ht]
\centering
\includegraphics[width=0.45\textwidth]{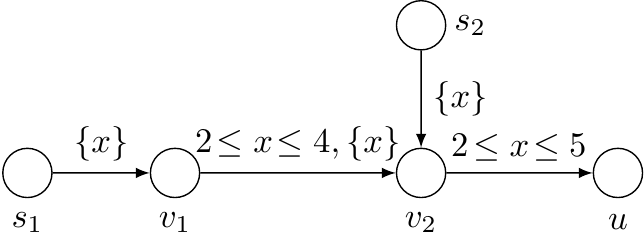} \ \hspace{.3cm}
\includegraphics[width=0.5\textwidth]{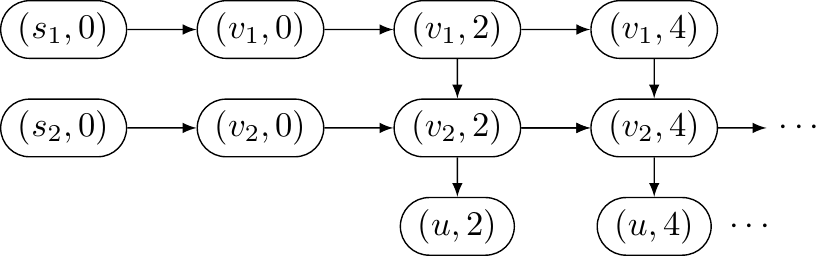}
\caption{\label{1c2ng} An attempt to translate a TNG to an NG.}
\end{figure}

Further, 
to see the difficulty in finding such a bound, consider, for example, a cost-sharing game in which all players, on their paths to their targets, need to stay for one time unit in a ``gateway'' vertex $v$ that costs $1$ (see details in Section~\ref{sec tb}). Assume also that, for $1 \leq i \leq k$, Player $i$ can only reach $v$ in times that are multiples of $p_i$, for relatively prime numbers $p_1,\ldots,p_k$. The SO is obtained when all players synchronize their visits to $v$, and such a synchronization forces them to wait till time $p_1 \cdot \ldots \cdot p_k$, which is exponential in the TNG.

The lack of an upper bound on the global time in TNGs
demonstrates that we need a different approach to obtain positive results for general TNGs. We show that TNGs are guaranteed to have an NE. Our proof uses a combination of techniques from real-time models and resource allocation games. Recall that a PTA assigns a price to a timed word. We are able to reduce the {\em best-response\/} and the {\em social-optimum\/} problems to and from 
the problem of finding cheapest runs in PTAs \cite{BBBR07}, showing that
the problems are PSPACE-complete. Next, we show that TNGs are potential games. Note that since players have 
uncountably many strategies, the fact that TNGs are potential games does not immediately imply existence of an NE, 
as a best-response sequence may not be finite.
We show that there is a 
best-response sequence
that terminates in an NE.
For this, we first need to show the existence of an integral best-response, which is obtained from the reduction to PTAs. Finally, given a TNG, we find a time $T$ such that there exists an NE in which all players reach their destination by time $T$.


\section{Preliminaries}
\label{sec:prelims}
\subsection{Resource allocation games and network games}
For $k \in \Nat$, let $[k]=\set{1,\ldots, k}$.
A {\em resource allocation game} (RAG, for short) is $R = \zug{k, E, \set{\Sigma_i}_{i \in [k]}, \allowbreak \set{\ell_e}_{e \in E}}$, where $k\in \Nat$ is the number of players; $E$ is a set of resources; for $i \in [k]$, the set strategies of Player~$i$ is $\Sigma_i \subseteq 2^E$; and, for $e \in E$, the {\em latency function} $\ell_e: [k] \rightarrow \Q_{\geq 0}$ maps a load on $e$ to its cost under this load. A {\em profile\/} is a choice of a strategy for each player. The set of profiles of $R$ is $\Prof(R) = \Sigma_1 \times \ldots \times \Sigma_k$. 
For $e \in E$, we define the {\em load} on $e$ in a profile $P=\zug{\sigma_1, \ldots, \sigma_k}$, denoted $load_P(e)$, as the number of players using $e$ in $P$, thus $load_P(e) = |\set{i \in [k]: e \in \sigma_i}|$. The cost a player pays in profile $P$, denoted $cost_i(P)$, depends on the choices of the other players. We define $cost_i(P) = \sum_{e \in \sigma_i} \ell_e(load_P(e))$.

{\em Network games} (NGs, for short) can be viewed as a special case of RAGs where strategies are succinctly represented by means of paths in graphs. An NG is $\N = \langle k, V$, $E$, $\set{\zug{s_i, u_i}}_{i \in [k]}, \set{\ell_e}_{e \in E}\rangle$, where $\zug{V, E}$ is a directed graph; for $i \in [k]$, the vertices $s_i$ and $u_i$ are the {\em source\/} and {\em target\/} vertices of Player~$i$; and the latency functions are as in RAGs. The set of strategies for Player~$i$ is the set of simple paths from $s_i$ to $u_i$ in $\N$. Thus, in NGs, the resources are the edges in the graph.

We distinguish between two types of latency functions.
In {\em cost-sharing games}, the players that visit a vertex share its cost equally. Formally, every $e \in E$ has a cost $c_e \in \Q_{\geq 0}$ and its latency function is $\ell_e(l) = \frac{c_e}{l}$. Note that these latency functions are decreasing, thus the load has a positive effect on the cost.
In contrast, in \emph{congestion games}, the cost functions are non-decreasing and so the load has a negative effect on the cost. Typically, the latency functions are restricted to simple functions such as linear latency functions, polynomials, and so forth. 

\subsection{Timed networks and timed network games}
A \emph{clock} is a variable that gets values from $\realpos$ and whose value increases as time elapses. A {\em reset\/} of a clock $x$ assigns value $0$ to $x$.
A \emph{guard} over a set $C$ of clocks
is a conjunction of \emph{clock constraints} of the form $x \sim m$, for $x \in C$, $\sim \, \in \{\le, =, \ge\}$, and $m \in \Nat$. Note that we disallow guards that use the operators $<$ and $>$
(see Remark~\ref{rem:<}).
A guard of the form $\bigwedge_{x \in C} x \ge 0$ is called \emph{true}.
The set of guards over $C$ is denoted $\Phi(C)$.
A \emph{clock valuation} is an assignment $\kappa: C \rightarrow \realpos$.
A clock valuation $\kappa$ \emph{satisfies} a guard $g$, denoted $\kappa \models g$,
if the expression obtained from $g$ by replacing each clock $x \in C$ with the value $\kappa(x)$ is valid. 
\stam{
\begin{remark}
This syntactic restriction of not allowing $<$ and $>$ in the guards is in order to focus on the main results of the paper without dealing with open sets, or notions of nearly optimal profile instead of an optimal profile, as we will see later, thus keeping the discussion simpler. 
\end{remark}
}
\stam{
For every $d \in \realpos$ and every clock valuation $\kappa$, we define the valuation $\kappa + d$ as $(\kappa + d)(x) = \kappa(x) + d$, for each $x \in C$.
For every subset $R$ of $C$, the valuation $\kappa_{[R \leftarrow \overline{0}]}$ is defined as
$\kappa_{[R \leftarrow \overline{0}]}(x)=0$ if $x \:\in\: R$, else $\kappa_{[R \leftarrow \overline{0}]}(x) = \kappa(x)$.
}

A {\em timed network\/} is a tuple
$\A = \zug{C,V, E}$,
where $C$ is a set of clocks, $V$ is a set of vertices, and
$E \subseteq V \times \Phi(C) \times 2^C \times  V$ is a set of directed edges in which each edge $e$ is associated with a guard $g \in \Phi(C)$ that should be satisfied when $e$ is traversed
and a set $R \subseteq C$ of clocks that are reset along the traversal of $e$.

When traversing a path in a timed network, time is spent in vertices, and edges are traversed instantaneously. 
Accordingly, 
a {\em timed path\/} in $\A$ is a sequence
$\eta=\zug{\tau_1, e_1}, \dots$, $\zug{\tau_n, e_n} \in (\realpos \times E)^*$, describing edges that the path traverses along with their traversal times. The timed path $\eta$ is {\em legal} if the edges are successive and the guards associated with them are satisfied. Formally, there is a sequence $\zug{v_0, t_{0}}, \dots, \zug{v_{n-1}, t_{n-1}},v_n \in (V \times \realpos)^*\cdot V$, describing the vertices that $\eta$ visits and the time spent in these vertices, such that for every $1 \leq j \leq n$, the following hold: (1)  
$t_{j-1}=\tau_j-\tau_{j-1}$, 
with $\tau_0=0$, (2) there is $g_j \in \Phi(C)$ and $R_j \subseteq C$, such that $e_j= \zug{v_{j-1}, g_j, R_j, v_j}$, (3) there is a clock valuation $\kappa_j$ that describes the values of the clocks 
before the
incoming edge to
vertex $v_j$ is traversed. Thus, $\kappa_1(x) = t_0$, for all $x \in C$, and for $1 < j \leq n$, we distinguish between clocks that are reset when $e_{j-1}$ is traversed and clocks that are not reset: for $x \in R_{j-1}$, we define $\kappa_j(x) = t_{j-1}$, and for $x \in (C \setminus R_{j-1})$, we define $\kappa_j(x) = \kappa_{j-1}(x) + t_{j-1}$, and (4) for every $1 \leq j \leq n$, we have that $\kappa_j \models g_j$. We sometimes refer to $\eta$ also as the sequence $\zug{v_0, t_{0}}, \dots, \zug{v_{n-1}, t_{n-1}},v_n$. 

Consider a finite set $T \subseteq \realpos$ of time points.
We say that a timed path $\eta$ is a {\em $T$-path\/}
if all edges in $\eta$ are taken at times in $T$.
Formally, for all 
$1\le j \le n$,
we have that $\tau_j \in T$. We refer to the time at which $\eta$ {\em ends} as the time $\tau_n$ at which the destination is reached. We say that $\eta$ is {\em integral} if $T \subseteq \Nat$.

A {\em timed network game\/} (TNG, for short) extends an NG by imposing constraints on the times at which edges may be traversed. Formally,  
$\T = \zug{k, C, V, E, \{\ell_v\}_{v \in V},$ $\zug{s_i, u_i}_{i \in [k]}}$ includes a set $C$ of clocks, and $\zug{C, V, E}$ is a timed network. Recall that while traversing a path in a timed network, time is spent in vertices.
Accordingly, the latency functions now apply to vertices, thus $\ell_v: [k] \rightarrow \Q_{\geq 0}$ maps a load on vertex $v$ to its cost under this load. 
Traversing an edge is instantaneous and is free of charge.
A strategy for Player~$i$, for $i \in [k]$, is then a legal timed path from $s_i$ to $u_i$. 
We assume all players have at least one strategy.

\begin{remark}
A possible extension of TNGs is to allow costs on edges. Since edges are traversed instantaneously, these costs would not be affected by load. Such an extension does not affect our results and we leave it out for sake of simplicity. Another possible extension is allowing strict time guards, which we discuss in Remark~\ref{rem:<}.
\end{remark}
\stam{
Note that in network games, cost is incurred due to sharing of resources.
In TNGs, the resources are the vertices used over a period of time and cost is incurred by spending time in the vertices.
Since in our model, 
no time elapses on the edges, we do not have cost on the edges.
Technically, however, all the positive results and the upper bounds presented in the paper hold even when we add costs to the edges and absence of cost on the edges strengthen the negative results -- all hold already for the model without costs on the edges.
The syntactic restriction of not allowing $<$ and $>$ in the guards is in order to focus on the main results of the paper without dealing with open sets and to keep the discussion simple.
\end{remark}
}

The cost Player~$i$ pays in profile $P$, denoted $cost_i(P)$, depends on the vertices in her timed path, the time spent on them, and the load during the visits. In order to define the cost formally, we need some definitions. 
For a finite set $T \subseteq \realpos$ of time points, we say that a timed path is a $T$-strategy if it is a $T$-path. Then, a profile  $P$ is a $T$-profile if it consists 
only 
of $T$-strategies. 
Let $t_{max} = \max(T)$. For $t \in T$ 
such that $t<t_{max}$, 
let $\nextT(t)$ be the minimal time point in $T$ that is strictly larger than $t$. We partition the interval $[0,t_{\max}]$ into a set $\Upsilon$ of sub-intervals
$[m,\nextT(m)]$ for every $m \in (T \cup \{0\}) \setminus \set{t_{max}}$.
We refer to the sub-intervals in $\Upsilon$ as {\em periods}.
Suppose $T$ is the minimal set such that $P$ is a $T$-profile. Note that $\Upsilon$ is the coarsest partition of $[0,t_{\max}]$ into periods such that no player crosses an edge within a period in $\Upsilon$. We denote this partition by $\Upsilon_P$.

For a player $i \in [k]$ and a period $\gamma \in \Upsilon_P$,
let $\visits_P(i,\gamma)$ be the vertex that Player $i$ visits during period $\gamma$.
That is, if $\pi_i=\zug{v^i_0, t^i_{0}}, \dots, \zug{v^i_{n_i-1}, t^i_{n_i-1}},v^i_{n_i}$ is a legal timed path that is a strategy for Player~$i$ and $\gamma = [m_1, m_2]$, then $\visits_P(i,\gamma)$ is the vertex $v_j^i$ for the index $1 \le j < n_i$ such that 
$\tau^i_{j}  \leq m_1 \leq m_2 \leq \tau^i_{j+1}$,
and $\visits_P(i,\gamma)$ is the vertex $v_0^i$ if 
$0 = m_1 \leq m_2 \leq \tau^i_1$.
Note that since $P$ is a $T$-profile, 
for each period $\gamma \in \Upsilon_P$,
the number of players that stay in each vertex $v$ during $\gamma$ is fixed.
Let $\load_P(v,\gamma)$ denote this number.
Formally $\load_P(v,\gamma) = |\{i: \visits_P(i,\gamma) = v\}|$.
Finally, for a period $\gamma = [m_1, m_2]$, let $|\gamma| = m_2 - m_1$ be the \emph{duration} of $\gamma$.
Suppose Player~$i$'s path ends at time $\tau^i$.
Let $\Upsilon^i_P \subseteq \Upsilon_P$ denote the periods that end by time $\tau_i$. 

Recall that the latency function
$\ell_v:[k] \longrightarrow \Q_{\ge 0}$ 
maps the number of players that simultaneously visit vertex $v$
to the price that each of them pays per time unit. If $\visits_P(i,\gamma)=v$,
then the cost of Player $i$ in $P$, over the period $\gamma$ is
$cost_{\gamma, i}(P) = \ell_v(\load_P(v,\gamma)) \cdot |\gamma|$. We define $cost_i(P) = \sum_{\gamma \in \Upsilon^i_P} cost_{\gamma, i}(P)$. The cost of the profile $P$, denoted $cost(P)$, is the total cost incurred by all the players,
i.e., $cost(P) = \sum_{i=1}^k cost_i(P)$.

A $T$-strategy is called an \emph{integral strategy} when $T \subseteq \Nat$,
and similarly for \emph{integral} profile.

A profile $P=\zug{\pi_1, \dots, \pi_k}$ is said to \emph{end} by time $\tau$
if for each $i \in [k]$, the strategy $\pi_i$ ends by time $\tau$.
Consider a TNG $\T$ that has a cycle such that a clock $x$ of $\T$ is reset on the cycle.
It is not difficult to see that this may lead to $\T$ having infinitely many integral profiles that end by different times.
A TNG $\T$ is called \emph{global} if it has a single clock $x$
that is never reset.
We use GTNG to indicate that a TNG is global.

As in RAGs, we distinguish between cost-sharing TNGs 
that have cost-sharing latency functions and congestion TNGs 
in which the latency functions are non-decreasing.


\subsection{Stability and efficiency} 
Consider a game $G$.
For a profile $P$ and a strategy $\pi$ of player $i \in [k]$,
let $P[i \leftarrow \pi]$ denote the profile obtained from $P$
by replacing the strategy of Player $i$ in $P$ by $\pi$.
A profile $P$ is said to be a {\em (pure) Nash equilibrium\/} (NE)
if none of the players in $[k]$
can benefit from a unilateral deviation from her
strategy in $P$ to another strategy.
Formally, for every Player~$i$ and every strategy $\pi$ for Player $i$,
it holds that $cost_i (P[i \leftarrow \pi]) \ge cost_i(P)$.
\stam{
Given a profile $P = \zug{\pi_1, \dots, \pi_k}$, a \emph{best response} (BR, for short) of a player, say Player~$k$, in profile $P$ is a strategy $\pi'_k$ such that for all strategies $\pi$ of Player~$k$,
we have that $cost_k(P[k \leftarrow \pi'_k]) \le cost_k(P[k \leftarrow \pi])$.
Given a profile $P = \zug{\pi_1, \dots, \pi_k}$ that is not an NE, a \emph{better response} of a player, say Player~$k$, from profile $P$ is a strategy $\pi'_k$ such that $cost_k(P[k \leftarrow \pi'_k]) < cost_k(P)$.
}

A {\em social optimum\/} (SO) of a game $G$ is
a profile that attains the infimum cost over all profiles.
We denote by $SO(G)$ the cost of an SO profile; i.e., $SO(G) = \inf_{P \in \Prof(G)} cost(P)$.
It is well known that decentralized decision-making may lead to sub-optimal solutions from the
point of view of the society as a whole. 
We quantify the inefficiency incurred due to self-interested behavior by the
\emph{price of anarchy} (PoA) \cite{KP09b,Pap01} and \emph{price of
stability} (PoS) \cite{AD+08} measures. The PoA is the worst-case inefficiency of a Nash equilibrium, while the PoS measures the best-case inefficiency of a Nash equilibrium. Note that unlike resource allocation games in which the set of profiles is finite, in TNGs there can be uncountably many NEs, so both PoS and PoA need to be defined using infimum/supremum rather than min/max. Formally,

\begin{definition}
Let $\mathcal{G}$ be a family of games, and let $G \in \G$ be a game in $\mathcal{G}$.
Let $\rm \Gamma(G)$ be the set of Nash equilibria of the game $G$. Assume that $\rm \Gamma(G)\neq \emptyset$.
\begin{itemize}
\item The {\em price of anarchy} of $G$ is 
$PoA(G) = \sup_{P\in \rm \Gamma(G)} cost(P)/SO(G)$.
The {\em price of anarchy} of the family of games $\mathcal{G}$
is $PoA(\mathcal{G}) = sup_{ G\in \mathcal{G}}PoA(G)$.
\item
The {\em price of stability} of $G$ is 
$PoS(G) = \inf_{P\in \rm \Gamma(G)} cost(P)/SO(G)$.
The {\em price of stability} of the family of games $\mathcal{G}$ is
$PoS(\mathcal{G}) = sup_{ G\in \mathcal{G}}PoS(G)$.
\end{itemize}
\end{definition}


\section{The Best-Response and the Social-Optimum Problems}
Consider a TNG $\T =  \zug{k, C, V, E, \{\ell_v\}_{v \in V}, \zug{s_i, u_i}_{i \in [k]}}$. In the {\em best-response problem} (BR problem, for short), we ask how a player reacts to a choice of strategies of the other players. Formally, let $\pi_1, \ldots, \pi_{k-1}$ be a choice of 
integral\footnote{We choose integral strategies since strategies with irrational times cannot be represented as part of the input; for strategies that use rational times, the best response problem can be solved with little modification in the proof of Theorem \ref{thm:BR-TNG-to-PTA}.}
strategies for Players~$1,\ldots,k-1$ in $\T$. 
We look for a strategy $\pi_k$ that minimizes $cost_k(\zug{\pi_1,\ldots, \pi_k})$. The choice of allowing Player~$k$ to react is arbitrary and is done for convenience of notation. In the {\em social optimum problem} (SOPT problem, for short), we seek a profile that maximizes the social welfare, or in other words, minimizes the sum of players' costs.

In this section we describe \emph{priced timed automata} (PTAs, for short) \cite{ALP01,BFHLPRV01} and show that while they are different from TNGs both in terms of the model and the questions asked on it, they offer a useful framework for reasoning about TNGs. In particular, we solve the BR and SOPT problems by 
reductions
to problems about PTAs.

\subsection{From TNGs to priced timed automata}
\label{TNGs to PTAs}
A PTA \cite{ALP01,BFHLPRV01} is $\P = \zug{C, V, E, \set{r_v}_{v \in V}}$, where $\zug{C, V, E}$ is a timed network and $r_v \in \Q_{\ge 0}$ 
is the {\em rate} of vertex $v \in V$. Intuitively, the rate $r_v$ specifies the cost of staying in $v$ for a duration of one time unit. Thus, a timed path 
$\eta = \zug{v_0, t_0}, \ldots, \zug{v_n,t_n},v_{n+1}$ in a PTA has a price, denoted $price(\eta)$, which is $\sum_{0 \leq j \leq n} r_v \cdot t_v$. The size of $\P$ is $|V| + |E|$ plus the number of bits needed in the binary encoding of the numbers appearing in guards and rates in $\P$.
\footnote{In general, PTAs have rates on transitions and strict time guards, which we do not need here.}

Consider a PTA $\P$ and two vertices $s$ and $u$.  Let $paths(s, u)$ be the set of timed paths from $s$ to $u$. We are interested in cheapest timed paths in $paths(s, u)$. 
A priori, there is no reason to assume that the minimal price is attained, thus we are interested in the optimal price, denoted $opt(s, u)$, which we define to be $\inf\set{price(\eta) : \eta \in paths(s, u)}$. The corresponding decision problem, called the \emph{cost optimal reachability problem} (COR, for short) takes in addition a threshold $\mu$, and the goal is to decide whether $opt(s, t) \leq \mu$. Recall that we do not allow the guards to use the operators $<$ and $>$.

\begin{theorem}
\label{thm:PTA}
\cite{BBBR07,FJ13} The COR problem is PSPACE-complete for PTAs with two or more clocks. 
Moreover, the optimal price is attained by an integral path, i.e., there is an integral path $\eta \in paths(s, u)$ with $price(\eta) = opt(s, u)$.
\end{theorem}

\stam{
\begin{remark} \label{remark PTA}
We can consider PTAs where the guards on the edges have clocks compared to rationals instead of integers.
Considering such a PTA $\P$ with rationals appearing in the guards, we can multiply them by the LCM of their denominators and for every vertex $v$, we divide the rate $r_v$ by the same and obtain a PTA $\P$ where the clocks are compared with integers in the guards and corresponding to every path $\eta$ in $\P$ there is a path $\eta'$ in $\P'$ such that $price(\eta) = price(\eta')$ and vice versa.
\end{remark}
}

In Sections~\ref{brp} and~\ref{soptp} below, we reduce problems on TNGs to problems on PTAs. The reductions allow us to obtain properties on strategies and profiles in TNGs using results on PTAs, which we later use in combination with techniques for NGs in order to solve problems on TNGs. 

\subsection{The best-response problem}
\label{brp}
\begin{theorem}
\label{thm:BR-TNG-to-PTA}
Consider a TNG $\T$ with $n$ clocks and 
integral
strategies $\pi_1, \ldots, \pi_{k-1}$ for Players~$1,\ldots,k-1$. There is a PTA $\P$ with $n+1$ clocks and two vertices $v$ and $u$ such that there is a one-to-one cost-preserving correspondence between strategies for  Player~$k$ in $\T$ and timed paths from $v$ to $u$: for every strategy $\pi_k$ in $\T$ and its  corresponding path $\eta$ in $\P$, we have $cost_k(\zug{\pi_1, \ldots, \pi_k}) = price(\eta)$. 
\end{theorem}
\begin{proof}
Consider a TNG $\T = \zug{k, V, E, C, \{\ell_v\}_{v \in V}, \zug{s_i, u_i}_{i \in [k]}}$, where $C = \{x_1, \dots, x_m\}$. Let $Q = \zug{\pi_1,\ldots, \pi_{k-1}}$ be a choice of timed paths for Players $1,\ldots,k-1$. 
Note that $Q$ can be seen as a profile in a game that is obtained from $\T$ by removing Player~$k$, and we use the definitions for profiles on $Q$ in the expected manner. 
Let $T \subseteq \Q$ be the minimal set of time points for which all the strategies in $Q$ are $T$-strategies. Consider two consecutive time points $a,b \in T$, i.e., there is no $c \in T$ with $a < c < b$. Then, there are players that cross edges at times $a$ and $b$, and no player crosses an edge at time points in the interval $(a,b)$. Moreover, let $t_{max}$ be the latest time in $T$, then $t_{max}$ is the latest time at which a player reaches her destination. Let $\Upsilon_Q$ be a partition of $[0,t_{max}]$ according to $T$. We obtain $\Upsilon'_Q$ from $\Upsilon_Q$ by adding the interval $[t_{\max}, \infty)$.

A key observation is that the load on all the vertices is unchanged during every interval in $\Upsilon'_Q$. For a vertex $v \in V$ and  $\delta \in \Upsilon_Q$, the cost Player~$k$ pays per unit time for using $v$ in the interval $\delta$ is $\ell_v(load_Q(v, \delta)+1)$. On the other hand, since all $k-1$ players reach their destination by time $t_{max}$, the load on $v$ after $t_{max}$ is $0$, and the cost Player~$k$ pays for using it then is $\ell_v(1)$. 

The PTA $\P$ that we construct has $|\Upsilon'_Q|$ copies of $\T$, thus its vertices are $V \times \Upsilon'_Q$.
Let $\delta_0 = [0, b] \in \Upsilon'_Q$ be the first interval. We consider paths from the vertex $v=\zug{s_k,\delta_0}$, which is the copy of Player~$k$'s source in the first copy of $\T$, to a target $u$, which is a new vertex we add and whose only incoming edges are from vertices of the form $\zug{u_k, \delta}$, namely, the copies of the target vertex $u_k$ of Player~$k$. We construct $\P$ such that each such path $\eta$ from $v$ to $u$ in $\P$ corresponds to a legal strategy $\pi_k$ for Player~$k$ in $\T$, and such that $cost_k(\zug{\pi_1,\ldots, \pi_{k-1}, \pi_k}) = price(\eta)$. The main difference between the copies are the vertices' costs, which depend on the load as in the above. We refer to the $n$ clocks in $\T$ as {\em local clocks}. In each copy of $\P$, we use the local clocks and their guards in $\T$ as well as an additional {\em global} clock that is never reset to keep track of global time. Let $\delta = [a,b] \in \Upsilon_Q$ and $\delta' = [b,c] \in \Upsilon'_Q$ be the following interval. Let $\T_\delta$ and $\T_{\delta'}$ be the copies of $\T$ that corresponds to the respective intervals. The local clocks guarantee that a path in $\T_\delta$ is a legal path in $\T$. The global clock allows us to make sure that (1) proceeding from $\T_\delta$ to $\T_{\delta'}$ can only occur precisely at time $b$, and (2) proceeding from 
$\zug{u_k,\delta}$
in $\T_{\delta}$ to the target $u$ can only occur at a time in the interval $\delta$.
We now formalize the intuition of the reduction given above.
We define $\P = \zug{V', E', C \cup \set{x_{n+1}}, \set{r_v}_{v \in V'}}$, where 
$V' = (V \times \Upsilon'_Q) \cup \{u\}$, 
and $E' = E'_l \cup E'_i \cup E'_t$, where
$E'_l$ is a set of \emph{external} edges, $E'_i$ is a set of \emph{internal} edges
and $E'_t$ is a set of \emph{target} edges.
Let $\delta_1, \dots, \delta_{|\Upsilon'_Q|}$ be the set of intervals arranged according to increasing $\inf(\delta_j)$, i.e., for all $j \in [|\Upsilon'_Q|-1]$, we have that $\inf(\delta_j) < \inf(\delta_{j+1})$.
Also for every interval $\delta=\delta_j$, we represent by $next(\delta)$, the interval $\delta_{j+1}$, and let $\tau_{\delta_j}=\sup(\delta_j)$.
For each $v \in V$, and $\delta \in \Upsilon_Q$, there is an {\em external} edge of the form $\zug{\zug{v,\delta}, \set{x_{n+1} = \tau_\delta}, \emptyset, \zug{v,next(\delta)}}$, where $\tau_\delta = \sup(\delta)$.
Hence an external edge moves from copy $\delta$ to its next copy at global time $\tau_\delta$. 
The {\em internal} edges in a copy match the ones in $\T$.
Thus, 
for every $\delta \in \Upsilon_Q$,
we have an edge $e' = \zug{\zug{v,\delta},g,R,\zug{v', \delta}}$ in $E'$ iff there is an edge $e = \zug{v, g, R, v'} \in E$. 
Also the guard and the clock reset on $e'$ are exactly the same as that of $e$.
Note that clock $x_{n+1}$ is not used in the internal edges.
For each copy corresponding to $\delta \in \Upsilon_Q$, there is a \emph{target} edge from the vertex $(u_k, \delta)$ to $u$
with the guard $x_{n+1} \leq \tau_\delta$, and from the copy corresponding to the last interval $\delta_{|\Upsilon'_Q|}$,
there is an edge from $(u_k, \delta_{|\Upsilon'_Q|})$ to $u$ with the guard $x_{n+1} \ge \tau_{\delta_{|\Upsilon'_Q|-1}}$.
Finally, we define the rate of a vertex $\zug{v,i}$.
Let $l$ be the load on $v$ during time interval $\delta$, then the rate of $\zug{v,\delta}$ is $\ell_v(l+1)$.
Note that for every vertex $\zug{v,\delta_{|\Upsilon'_Q|}}$, in the $|\Upsilon'_Q|$-th copy of $\P$, the rate is $\ell_v(1)$.

We prove that the cost of the best response strategy of Player~$k$ in $\T$ is the same
as the cost of a \emph{cost optimal path} in $\P$.
We consider a strategy $\pi$ of Player~$k$ in $\T$ and let $P = \zug{\pi_1, \dots, \pi_k}$
and show that there is a path $\eta$ in $\P$ such that $cost_k(P)$ is the same as cost of the path $\eta$.
Similarly, for a path $\eta$ in $\P$, we show that there exists a strategy $\pi$ of Player~$k$ such that
again $cost_k(P)$ is the same as cost of the path $\eta$ in $\P$.

Consider the strategy $\pi = (v_1, t_1), \dots, (v_p, t_p),u_k$ of Player~$k$ in $\T$ such that $v_1 = s_k$.
We construct a timed path $\pi' = \zug{\zug{v_1,\delta_{1_1}}, t_1}, \zug{\zug{v_2,\delta_{1_1}}, t_2},$ $\ldots, \zug{\zug{v_\ell,\delta_\ell}, t_\ell}$, $\zug{\zug{v_{\ell+1}$, $\delta_{\ell+1}},0}$, $u$ in $\P$ in the following manner.
Firstly, $v_1 = s_k$, $v_{\ell+1} = u_k$ and $i_1 = 1$.
Consider the mapping $g: \realpos \mapsto \Upsilon'_Q$ such that
$g(0) = \delta_1$, and 
for $\sum_{j=1}^i t_j \leq \tau_1$, we have $g(\sum_{j=1}^i t_j) =\delta_1$ and when $\sum_{j=1}^i t_j \geq \tau_1$, for all $1 \le i \le \ell$, we have
$g(\sum_{j=1}^i t_j) = \delta_\xi$ if $\tau_{\xi-1} \le g(\sum_{j=1}^{i-1}t_j) \leq \tau_\xi$.
Intuitively, $g$ maps a global time $\tau$ to the $\xi^{th}$ copy of the TNG $\T$ in $\P$
if $\tau_{\xi-1} \le \tau \leq \tau_\xi$.
Consider $(v_i, t_i)$ in the path $\pi$ for some $i \in [p]$.
Suppose there are times $\tau_r, \tau_{r+1}, \dots, \tau_{r+h} \in T$
such that all these $h+1$ times belong to the interval $[\sum_{j=1}^{i-1} t_j, \sum_{j=1}^i t_j]$.
Then in $\pi'$, we replace $(v_i, t_i)$ by
$\zug{\zug{v_i, g(\sum_{j=1}^{i-1} t_j)}, \tau_r - \sum_{j=1}^{i-1}t_j}$,
$\zug{\zug{v_i, g(\tau_r)}, \tau_{r+1} - \tau_r}$, $\dots$,
$\zug{\zug{v_i, g(\tau_{r+l})}, \sum_{j=1}^i t_j - \tau_{r+l}}$.
Again from the definition of $f_v$, for $v \in V'$, we can see that
the cost of the timed path $\pi'$ in $\P$ is the same as $cost_k(P)$.

Now we consider the other direction.
Consider a path $\eta = \zug{\zug{v_1,\delta_1}, t_1}, \zug{\zug{v_2,\delta_2}, t_2}$, $\ldots$, $\zug{\zug{v_\ell,\delta_\ell}, t_\ell},$ $\zug{\zug{v_{\ell+1}, \delta_{\ell+1}}, 0}, u$ in $\P$ such that $v_1 = s_k$ and $v_{\ell+1} = u_k$.
We construct a path $\pi$ in $\T$ from $\pi$ as follows.
Every internal edge $(\zug{\zug{v_j,\delta_j}, t_j}, \zug{\zug{v_{j+1},\delta_{j+1}}, t_{j+1}})$ is replaced by
the edge ($\zug{v_j, t_j}, \zug{v_{j+1}, t_{j+1}}$) in $\pi$.
Note that here $i_j = i_{j+1}$.
A sequence of external edges $(\zug{\zug{v_j,i_j}, t_j}, \dots, \zug{\zug{v_{j+l},i_{j+1}}, t_{j+l}})$ 
such that $v_j = v_{j+1} = \dots = v_{j+l}$
is replaced by $(\zug{v_j, t_j+t_{j+1} + \dots +t_{j+l}})$.
Let $t_0 = 0$.
The cost along path $\eta$ in $\P$ is $r_{v_1}(l_1) \cdot t_1 + r_{v_2}(l_2) \cdot t_2 + \dots +
r_{v_\ell}(l_\ell) \cdot t_\ell$, where for $1 \le j \le \ell$,
we have $l_j = \load_P(v_j, [\sum_{q=1}^{j}t_q - \sum_{q=1}^{j-1}t_q])$.
Let $P = \zug{\pi_1, \dots, \pi_{k-1}, \pi}$ is the profile obtained from the strategies
of the $k$ players.
From the definition of $f_v$ for each $v \in V'$, it is not difficult to see that
$cost_k(P)$ is the same as the cost of the timed path $\eta$.

Note that given integral strategies of $k-1$ players, an integral path in $\P$ translates to an integral strategy of Player~$k$ in $\T$.
Since it is known that a cost optimal path in $\P$ can be an integral path \cite{BBBR07},
the best response of Player~$k$ in $\T$ is an integral strategy.
\end{proof}

We conclude with the computational complexity of the BR problem. The decision-problem variant gets as input a TNG $\T$, 
integral
strategies $\pi_1, \ldots, \pi_{k-1}$ for Players~$1,\ldots$, $k-1$, and a value $\mu$, and the goal is to decide whether Player~$k$ has a strategy $\pi_k$ such that $cost_k(\zug{\pi_1,\ldots, \pi_k}) \leq \mu$. Theorem~\ref{thm:BR-TNG-to-PTA} implies a reduction from the BR problem to the COR problem and a reduction in the other direction is easy since PTAs can be seen as TNGs with a single player.
For one-clock instances, we show that the BR problem is NP-hard by a reduction from the subset-sum problem. Note the contrast with the COR problem in one-clock instances, which is NLOGSPACE-complete \cite{LMS04}. 

\begin{theorem}
\label{thm:BR-complexity}
The BR problem is PSPACE-complete for TNGs with two or more clocks. 
For one-clock cost-sharing and congestion TNGs it is in PSPACE and NP-hard.
\end{theorem}
\begin{proof}
We reduce the BR problem to and from the COR problem, which is PSPACE-complete for PTAs with at least two clocks \cite{BBBR07}. A PTA can be seen as a one-player TNG, thus the BR problem for TNGs with two or more clocks is PSPACE-hard. For the upper bound, given a TNG $\T$, strategies $Q = \zug{\pi_1,\ldots, \pi_{k-1}}$ for Players~$1,\ldots,k-1$, and a threshold $\mu$, we construct a PTA $\P$ as in the proof of Theorem~\ref{thm:BR-TNG-to-PTA}.
Note that the size of $\P$ is polynomial in the size of the input and that $\P$ has one more clock than $\T$. An optimal path in $\P$ is a best response for Player~$k$, and such a path can be found in PSPACE.

The final case to consider is TNGs with one clock. We show that the BR problem is NP-hard for such instances using a reduction from the {\em subset-sum} problem. The input to that problem is a set of natural numbers $A = \{a_1, \dots ,a_n\}$ and $\mu \in \Nat$, and the goal is to decide whether there is a subset of $A$ whose sum is $\mu$. 
We start with the cost-sharing case. The game we construct is a two-player game on a network that is depicted in Figure~\ref{fig-TNG1clk}. \PT has a unique strategy that visits vertex $v_{n+1}$ in the time interval $[\mu,\mu+1]$. A \PO strategy $\pi$ corresponds to a choice of a subset of $A$. \PO's source is $v_1$ and 
her
target is $u_2$. The vertex $v_{n+1}$ is the only vertex that has a cost, which is $1$, and the other vertices cost $0$. For $1 \leq i \leq n$, \PO needs to choose between staying in vertex $v_i$ for a duration of $a_i$ time units, and exiting the vertex 
through
the top edge, or staying $0$ time units, and exiting the vertex 
through
the bottom edge. Finally, 
she
must stay in $v_{n+1}$ for exactly one time unit. The cost \PO pays for $v_{n+1}$ depends on the load. If 
she
stays there in the global time interval $[\mu,\mu+1]$, 
she
pays $1/2$, and otherwise 
she
pays $1$. Thus, \PO has a strategy with which 
she
pays $1/2$ iff there is a subset of $A$ whose sum is $\mu$, and we are done.
\begin{figure}[ht]
\centering
\includegraphics[height=2.7cm]{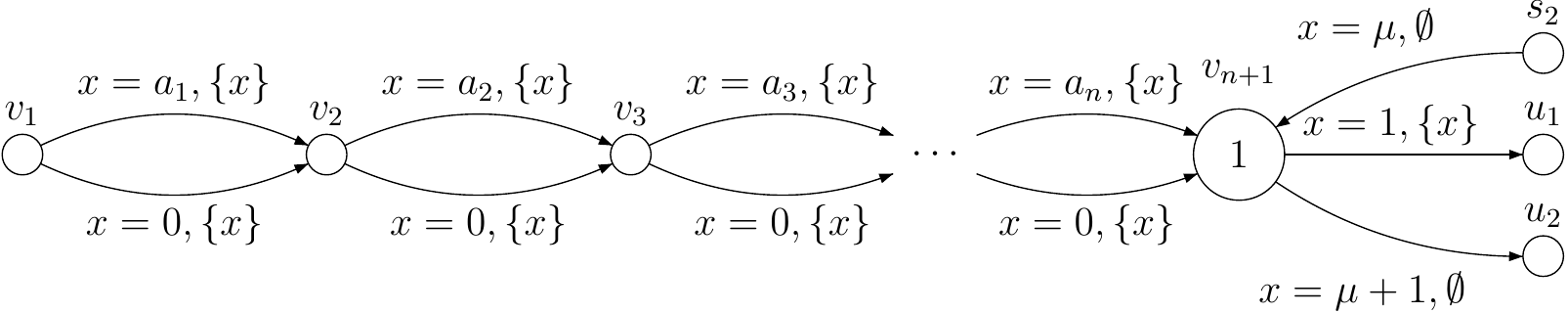}
\caption{\label{fig-TNG1clk} NP-hardness proof of best response problem in one clock TNG}
\end{figure}

The reduction for congestion games is similar. Recall that in congestion games, the cost increases with the load, thus a player would aim at using a vertex together with as few other players as possible. The network is the same as the one used above. Instead of two players, we use three players, where Players~$2$ and~$3$ have a unique strategy each. \PT must stay in $v_{n+1}$ in the time interval $[0,\mu]$ and Player~$3$ must stay there 
during the interval
$[\mu+1,\sum_{1 \leq i \leq n}a_i]$. As in the above, \PO has a strategy in which 
she
uses $v_{n+1}$ alone in the time interval $[\mu,\mu+1]$ iff there is a subset of $A$ whose sum is $\mu$.
\end{proof}

\subsection{The social-optimum problem}
\label{soptp}
\begin{theorem}
\label{thm:SO-TNG-to-PTA}
Consider a TNG $\T =  \zug{k, C, V, E, \{\ell_v\}_{v \in V}, \zug{s_i, u_i}_{i \in [k]}}$. There is a PTA $\P$ with $k \cdot |C|$ clocks, $|V|^k$ vertices, and two vertices $\bar{s}$ and $\bar{u}$ such that there is a one-to-one cost-preserving correspondence between profiles in $\T$ and paths from $\bar{s}$ to $\bar{u}$; namely, for a profile $P$ and its  corresponding path $\eta_P$, we have $cost(P) = price(\eta_P)$. 
\end{theorem}
\begin{proof}
First, we show how to construct, given a TNG $\T$ with self loops, a TNG $\T'$ that has no self loops. 
Consider a vertex $v$ that has a self loop $e = \zug{v, g, R, v}$
in $\T$.
In $\T'$, we remove $e$, we add a vertex $v'$,
a new clock $x$
and ``redirect'' $e$ to $v'$ while resetting $x$. Formally, we have the edge $\zug{v, g, R\cup \set{x},v'}$. We enforce that $v'$ is left instantaneously using an edge $\zug{v', \set{x=0}, \emptyset, v}$. Clearly, the strategies of the players in $\T$ and $\T'$ coincide.

Recall that the social optimum is obtained when the players do not act selfishly, rather they cooperate to find the profile that minimizes their sum of costs. Let $\T = \zug{k, C, V, E, \{\ell_v\}_{v \in V},\zug{s_i, u_i}_{i \in [k]}}$. We construct a PTA $\P$ by taking $k$ copies of $\T$. For $i\in [k]$, the $i$-th copy is used to keep track of the timed path that Player~$i$ uses. We need $k$ copies of the clocks of 
$\T$
to guarantee that the individual paths are legal. 
Recall that the players' goal is to minimize their total cost, thus for each point in time, the price they pay in $\P$ is the sum of their individual costs in $\T$. More formally, consider a vertex $\bar{v} = \zug{v_1,\ldots, v_k}$ in $\P$ and let $S_{\bar{v}} \subseteq V$ be the set of vertices that appear in $\bar{v}$. Then, the load on a vertex $v \in S_{\bar{v}}$ in $\bar{v}$ is $load_{\bar{v}}(v) = |\set{i: v_i=v}|$, and the rate of $\bar{v}$ is $\sum_{v \in S_{\bar{v}}} \ell_v(load_{\bar{v}}(v))$. 
We show below that the cost of the social optimum in $\T$ coincides with the price of the optimal timed path in $\P$ from $\zug{s_1,\ldots,s_k}$ to the vertex $\zug{u_1,\ldots,u_k}$, i.e., the vertices that respectively correspond to the sources and targets of all players.
Towards this,
we show that for every path $\eta$ in the PTA $\P$, there exists a profile $P_\eta$ in $\T$ such that $price(\eta)$ in $\P$ equals $cost(P_\eta)$ in $\T$.
Consider a timed path $\eta = (\bar{v}_1, {t}_1), \dots, (\bar{v}_n, {t}_n), \bar{u}$ in $\P$ where $\bar{s} = \bar{v_1}$ and $\bar{v}_j = \zug{v_j^1, \dots, v_j^k}$, for $1 \leq j \leq n$.
For each Player~$i \in [k]$, we construct a timed path $\pi_i$ in $\T$ as follows.
Intuitively, we restrict $\eta$ to the $i$-th component and remove recurring vertices. 
Consider an index $1 \leq j \leq n$ in which Player~$i$ changes vertices in $\eta$, thus $v^i_j \neq v^i_{j+1}$.  
We call such an index $j$ a {\em changing index}. 
Let $j_1,\ldots, j_m$ be the changing indices for Player~$i$. 
Let $j_0=1$ and $j_{m+1}=n$.
Note that between changing indices Player~$i$ does not change vertices. 
We define 
$\pi_i 
= \zug{v^i_{j_0}, t^i_0},\zug{v^i_{j_1}, t^i_1},\ldots, \zug{v^i_{j_m}, t^i_m}$, 
where for $p \in \{0\} \cup [m]$, we have $t^i_p = \sum_{j_{p} \leq l < j_{p+1}} t_l$. That is, $t^i_p$ is the total time that Player~$i$ spends in $v^i_{j_p}$ in $\eta$. 

We claim that $\pi_i$ is a legal path,
i.e., 
we claim that the respective guards are satisfied
when the path switches between vertices by crossing the edges.
Consider an index $1 \leq l \leq m$. We can prove by induction on the length of $\pi_i$ that 
for every clock $x \in C$, the value of $x$
after the prefix $\zug{v^i_{j_1}, t^i_1},\ldots, \zug{v^i_{j_l}, t^i_l}$ of $\pi_i$ equals the value of the clock $\zug{x,i}$ after the prefix $\zug{\bar{v}_1,t_1},\ldots,\zug{\bar{v}_{j_l}, t_{j_l}}$ of $\eta$. 
We note that from the way we construct the set $R$ of clocks that are reset along an edge of the PTA $\P$,
clock $\zug{x,i}$ can only be reset 
exactly at the times at which $x$ is reset in $\pi_i$. Finally, let $P_\eta = \zug{\pi_1,\ldots, \pi_k}$. It is not hard to see that $cost(P_\eta) = price(\eta)$. 

Showing correctness of the other direction of the reduction is dual; namely, given a profile $P = \zug{\pi_1,\ldots, \pi_k}$ in 
$\T$
we construct a path $\eta$ from $\bar{s}$ to $\bar{u}$ in 
$\P$
such that 
$cost(P) = price(\eta)$. 
\end{proof}


We turn to study the complexity of the SOPT problem. In the decision-problem variant, we are given a TNG $\T$ and a value $\mu$ and the goal is to decide whether there is a profile $P$ in $\T$ with $cost(P) \leq \mu$. 
Theorem~\ref{thm:SO-TNG-to-PTA} implies a reduction from the SOPT problem to the COR problem, and, as in the BR problem, the other direction is trivial. For one-clock instances, we use the same NP-hardness proof as in the BR problem. 

\begin{theorem}
\label{thm:SOPT-complexity}
The SOPT problem is PSPACE-complete for at least two clocks and it is NP-hard for TNGs with one clock.
\end{theorem}
\begin{proof}
In \cite{BBBR07}, the COR problem was shown to be in PSPACE as follows. Given a PTA $\P = \zug{C, V, E, \set{r_v}_{v \in V}}$, they construct a \emph{weighted discrete graph}, which is an extension of the region graph \cite{AD94}. 
Let $R_\P$ denote the region graph of $\P$. Then, the size of the weighted directed graph of $\P$ is at most  $(|C|+1) \cdot |R_\P|$. The size of the region graph is $\mathcal{O}(\chi^{|C|} \cdot |C|! \cdot |V|)$, where $\chi$ is the maximum constant appearing in the guards on the edges of $\P$.
Note that the size of this weighted discrete graph is exponential in the size of $\P$ and the length of a cost optimal path is bounded by the number of vertices in the weighted discrete graph.
Further the encoding of the cost of a cost optimal path can be done in PSPACE.
Hence one can decide in NPSPACE, thus in PSPACE, the COR problem.

Given a TNG $\T = \zug{k, C, V, E, \{\ell_v\}_{v \in V}, \zug{s_i, u_i}_{i \in [k]}}$, let $\P$ be the PTA constructed in the proof of Theorem \ref{thm:SO-TNG-to-PTA}.
Recall that $\P$ has $k|C|$ clocks and $|V|^k$ vertices and an SO in $\T$ corresponds to an optimal path $\eta$ in $\P$ and the cost of the SO profile equals the cost of the path $\eta$.
We note that the size of the weighted discrete graph for $\P$ is exponential in the size of the input $\T$.
Since the cost of an SO in $\T$ equals the cost of a cost optimal path in $\P$, the PSPACE-membership of the SOPT problem follows.

It remains to prove that the SOPT problem is PSPACE-hard.
In \cite{AL02}, the reachability problem for timed automata (which we call here timed networks) with only equality operators on the guards, has been shown to be PSPACE-complete.
Given a timed network $\mathcal{A} = \zug{C,V,E}$, an initial and final vertex $s$ and $u$ in $V$ respectively, we construct a one player TNG 
$\T = \zug{1, C, V', E', \{\ell_v\}_{v \in V}, \zug{s_1, u_1}}$, where $V'= V \cup \{s', u', v\}$ such that $s', u'$ and $v'$ are not in $V$.
The set $E'= E \cup \{(s',s),(u,u'),(s',v),(v,u')\}$.
All the vertices apart from $v$ are free vertices while $\ell_v(1) = 1$
such that the only player has a strategy that costs 
less than or equal to $0$
iff there is a path from $s$ to $u$ in $\A$. Note that in a one-player game, the social optimum is the strategy 
(which is a timed path)
that minimizes the cost of the player. 
The vertices $s'$ and $u'$ are the fresh initial and target vertices respectively
and connect 
$s'$ to $s$ and $u'$ to $u$.
Thus, if there is a path from $s$ to $u$
in $\mathcal{A}$,
then \PO has a strategy 
with cost $0$.
Otherwise, if there is no path from $s$ to $u$ in $\mathcal{A}$, then the cost of the strategy that uses $v$ is $1$.

For the lower bound of one-clock instances, we revisit the NP-hardness proof of the BR-problem in Theorem \ref{thm:BR-complexity}.
For a cost-sharing TNG, we note that an SO with cost no more than $1$ exists iff there is a solution to the subset-sum problem.
The proof for congestion TNG is analogous.
\end{proof}

\section{Existence of a Nash Equilibrium} 
\label{sec-existsNE}
The first question that arises in the context of games is the existence of an NE. 
In \cite{AGK17}, we showed that GTNGs are guaranteed to have an NE by reducing every GTNG to an NG.
We strengthen the result by showing that every TNG has an NE.

In order to prove existence, we combine techniques from NGs and 
use the reduction to PTA 
in Theorem \ref{thm:BR-TNG-to-PTA}.
A standard method for finding an NE is  
showing that a {\em best-response sequence} converges:
Starting from some profile $P = \zug{\pi_1, \ldots, \pi_k}$, one searches for a player that can benefit from a unilateral deviation. If no such player exists, then $P$ is an NE and we are done. Otherwise, let $\pi'_i$ be a beneficial deviation for Player~$i$, i.e., $cost_i(P) > cost_i(P[i \gets \pi'_i])$. The profile $P[i \gets \pi'_i]$ is considered next and the above procedure repeats.

A {\em potential function} for a game is a function $\Psi$ that maps profiles to costs, such that the following holds: for every profile $P = \zug{\pi_1, \ldots, \pi_k}$, $i \in [k]$, and strategy $\pi'_i$ for Player~$i$, we have $\Psi(P) - \Psi(P[i \gets \pi'_i]) = cost_i(P) - cost_i(P[i \gets \pi'_i])$,
i.e., the change in potential equals the change in cost of the deviating player.
A game is a {\em potential game} if it has a potential function. 
In a potential game with finitely many profiles, since the potential of every profile is non-negative and in every step of a best-response sequence the potential strictly decreases,
every best-response sequence terminates in an NE.
It is well-known that RAGs are potential games \cite{Ros73} and since they are finite, this implies that an NE always exists.

The idea of our proof is as follows. First, we show that TNGs are potential games, which does not imply existence of NE since TNGs have infinitely many profiles. Then, we focus on a specific best-response sequence   that starts from an integral profile and allows the players to deviate only to integral strategies.
Finally, we define {\em normalized TNGs} and show how to {\em normalize} a TNG in a way that preserves existence of NE. 
For normalized TNGs, we show that the potential reduces at least by $1$ along each step in the 
best-response sequence, thus it converges to an NE.

\begin{theorem}
\label{thm:potential-games}
TNGs are potential games.
\end{theorem}
\begin{proof}
Consider a TNG $\T = \zug{k, C, V, E, \{\ell_v\}_{v \in V}, \zug{s_i, u_i}_{i \in [k]}}$. Recall that for a profile $P$, the set of intervals that are used in $P$ is $\Upsilon_P$. We define a potential function $\Psi$ that is an adaptation of Rosenthal's potential function \cite{Ros73} to TNGs. We decompose the definition of $\Psi$ into smaller components, which will be helpful later on. For 
every
$\gamma \in \Upsilon_P$ and $v \in V$, we define 
$\Psi_{\gamma,v}(P) = \sum_{j=1}^{load_P(v,\gamma)} |\gamma| \cdot \ell_v(j)$, that is, we take the sum of $|\gamma| \cdot \ell_v(j)$ for all $j \in [load_P(v, \gamma)]$.
We 
define $\Psi_\gamma(P) = \sum_{v \in V} \Psi_{\gamma, v}(P)$, and we define $\Psi(P) = \sum_{\gamma \in \Upsilon_P} \Psi_\gamma(P)$. 
Let 
for some $i \in [k]$,
we have
$P'$ 
to
be a profile that is obtained by an unilateral deviation of Player~$i$
to a strictly beneficial strategy $\pi_i'$ from her current strategy in $P$, that is $P'= P[i \leftarrow \pi']$
for some $i \in [k]$. 
We show that $\Psi(P) - \Psi(P') = cost_i(P) - cost_i(P')$. 
Let $T$ and $T'$ be the minimal sets such that $P$ and $P'$ are $T$- and $T'$-profiles, respectively. Let $\Upsilon_{P,P'}$ be a set of intervals that refine the intervals in $\Upsilon_P$ and $\Upsilon_{P'}$ according to $T \cup T'$. Formally, consider an interval $[a,b] \in \Upsilon_P$. The definition of $\Upsilon_P$ implies that there is no time point $t \in T$ with $a < t < b$. On the other hand, if there exist time points $t_1,\ldots,t_n \in T'$ with $a < t_1 < \ldots<t_n < b$, then $[a,t_1],[t_1,t_2],\ldots,[t_n,b] \in \Upsilon_{P,P'}$, and dually for intervals in $\Upsilon_{P'}$. 

It is not hard to see that since $\Upsilon_{P,P'}$ refines $\Upsilon_P$ and $\Upsilon_{P'}$, we have $\sum_{\gamma \in \Upsilon_{P,P'}} \Psi_\gamma(P) = \Psi(P)$ and 
$\sum_{\gamma \in \Upsilon_{P,P'}} \Psi_\gamma(P') = \Psi(P')$. 
Consider an interval $\gamma \in \Upsilon_{P,P'}$. Recall that $P'$ is obtained by letting Player~$i$ change her strategy from the one in $P$ and the other players' strategies remain the same. Let $v = \visits_P(i, \gamma)$ and $v' = \visits_{P'}(i, \gamma)$. 
For every $v'' \neq v, v'$, since the other players do not change their strategies, the loads stay the same over the duration $\gamma$ and we have $\Psi_{\gamma,v''}(P) = \Psi_{\gamma,v''}(P')$. 
We consider the case where $v \neq v'$. Thus, Player~$i$ uses $v$ in the interval $\gamma$ in $P$ and uses $v'$ in the same interval in $P'$, or the other way around. Thus, we have 
$|load_{P}(v, \gamma) - load_{P'}(v, \gamma)| = 1$ and 
$|load_{P}(v', \gamma) - load_{P'}(v', \gamma)| = 1$. 
Suppose $load_P(v, \gamma) = load_{P'}(v,\gamma) + 1$ and 
$load_{P}(v', \gamma) = load_{P'}(v', \gamma) - 1$. 
Thus $\Psi_\gamma(P) - \Psi_\gamma(P') = (\Psi_{\gamma,v}(P) + \Psi_{\gamma,v'}(P)) - (\Psi_{\gamma,v}(P') + \Psi_{\gamma,v'}(P')) = (\Psi_{\gamma,v}(P) - \Psi_{\gamma,v}(P')) + (\Psi_{\gamma,v'}(P) - \Psi_{\gamma,v'}(P'))$ $= |\gamma| \cdot \ell_v(load_{P}(v, \gamma)) - |\gamma| \cdot \ell_{v'}(load_{P'}(v', \gamma))$.
Now the costs of Player $i$ in profile $P$ and $P'$ over the duration $\gamma$ are $|\gamma| \cdot \ell_v(load_{P}(v, \gamma))$ and $|\gamma| \cdot \ell_{v'}(load_{P'}(v', \gamma))$ respectively.
Thus $\Psi_\gamma(P) - \Psi_\gamma(P') = cost_{\gamma,i}(P) - cost_{\gamma,i}(P')$.
Since we sum up for all $\gamma \in \Upsilon_{P,P'}$, we get $\Psi(P) - \Psi(P') = cost_i(P) - cost_i(P')$, and we are done.
\end{proof}


Recall from Theorem \ref{thm:BR-TNG-to-PTA}, that given a TNG, a profile $P$ and an index $i$, we find the best response of Player $i$ by constructing a PTA.
If $P$ is an integral profile, from Theorem \ref{thm:PTA}, we have that the best response of Player $i$ also leads to an integer profile.
Thus we have the following lemma.
\begin{lemma} 
\label{lem:integral-BR}
Consider a TNG $\T$ and an integral profile $P$. For $i \in [k]$, if Player~$i$ has a beneficial deviation from $P$, then she has an integral beneficial deviation.
\end{lemma}

The last ingredient of the proof 
gives a lower bound for
the difference in cost that is achieved in a beneficial integral deviation
for some player $i \in [k]$,
which in turn bounds the change in potential.

We first need to introduce a \emph{normalized} form of TNGs. Recall that the latency function in a TNG $\T$ is of the form $\ell_v: [k] \rightarrow \Q_{\geq 0}$. In a normalized TNG all the latency functions map loads to natural numbers, thus for every vertex $v \in V$, we have $\ell_v: [k] \rightarrow \Nat$. Constructing a normalized TNG from a TNG is easy. Let $L$ be the least common multiple of the denominators of the elements in the set $\set{\ell_v(l): v \in V \mbox{ and } l \in [k]}$. For every latency function $\ell_v$ 
and every $l \in [k]$ ,
we construct a new latency function $\ell'_v$ by $\ell'_v(l) = \ell_v(l) \cdot L$. 

Consider a TNG $\T$ and let $\T'$ be the normalized TNG that is constructed from $\T$. It is not hard to see that for every profile $P$ and $i \in [k]$, we have $cost_i(P)$ in $\T'$ is $L \cdot cost_i(P)$ in $\T$. We can thus restrict attention to normalized TNGs as the existence of NE and convergence of best-response sequence in $\T'$ implies the same properties in $\T$. In order to show that a best-response sequence converges in TNGs, we bound the change of potential in each best-response step by observing that in normalized TNGs, the cost a player pays is an integer.

\begin{lemma} 
\label{lem-costBRCSCON}
Let $\T$ be a normalized TNG, $P = \zug{\pi_1, \ldots, \pi_k}$ be an integral profile in $\T$, and $\pi'_i$ be a beneficial integral deviation for Player~$i$, for some $i \in [k]$. Then,  $cost_i(P) - cost_i(P[i \gets \pi'_i]) \ge 1$.
\end{lemma}

We can now prove the main result in this section.

\begin{theorem} \label{thm-all have bne}
Every TNG has an integral NE. Moreover, from an integral profile $P$, there is a best-response sequence  that converges to an integral NE.
\end{theorem}
\begin{proof}
Lemma~\ref{lem:integral-BR} allows us to restrict attention to integral deviations. Indeed, consider an integral profile $P$. 
Lemma~\ref{lem:integral-BR}
implies that if no player has a beneficial integral deviation from $P$, then $P$ is an NE in $\T$. 
We start best-response sequence  from some integral profile $P_I$ and allow the players to deviate with integral strategies only. Consider a profile $P$ and let $P'$ be a profile that is obtained from $P$ by a deviation of Player~$i$. Recall 
from Theorem \ref{thm:potential-games}
that $cost_i(P) - cost_i(P') = \Psi(P) - \Psi(P')$. Lemma~\ref{lem-costBRCSCON} implies that when the deviation is beneficial, we have $\Psi(P) - \Psi(P') \geq 1$. Since the potential is 
non-negative,
the best-response sequence  above converges within $\Psi(P_I)$ steps.
\end{proof}
\vspace{-.2cm}
\begin{remark}
\label{rem:<}
A TNG that allows $<$ and $>$ operators on the guards is not guaranteed to have an NE. 
Indeed, in a PTA, which can be seen as a one-player TNG, strict guards imply that an optimal timed path may not be achieved. In turn, this means that an NE does not exist.
To overcome this issue, we use $\epsilon$-NE, for $\epsilon >0$; an $\epsilon$-deviation is one that improves the payoff of a player at least by $\epsilon$, and an $\epsilon$-NE is a profile in which no player has a $\epsilon$-deviation. Our techniques can be adapted to show that $\epsilon$-NE exist in TNGs with strict guards. The proof uses  the results of \cite{BBBR07} that show that an $\epsilon$-optimal timed path exists in PTAs. The proof technique for existence of NE in TNGs with non-strict guards can then be adapted to the strict-guard case.
\end{remark}

\stam{
TODO
\begin{itemize}
\item Check: Given a $T$-profile $P$,  does there exists a $c$ that depends on $T$ and the TNG such that $cost_i(P) - cost_i(P[i \gets \pi]) >= c$? This will allow us to strengthen the theorem to: from {\em every} profile (not necessarily integral), there is a best-response sequence  that converges to an NE.
\item Consider a non-normalized TNG. Bound the number of steps needed for convergence.
\end{itemize}
}

\section{Equilibrium Inefficiency}
In this section we address the problem of measuring the degradation in social welfare due to selfish behavior, which is measured by the PoS and PoA measures. 
We show that the upper bounds from RAGs on these two measures apply to TNGs. For cost-sharing TNGs, we show that the PoS and PoA are at most $\log k$ and $k$, respectively, as it is in cost-sharing RAGs. Matching lower bounds were given in \cite{AGK17} already for GTNGs.
For congestion TNGs with affine latency functions, we show that the PoS and PoA are $1 + \sqrt(3) / 3 \approx 1.577$ and $\frac{5}{2}$, respectively, as it is in congestion RAGs. Again, a matching lower bound for PoA is shown in \cite{AGK17} for GTNGs, and a matching lower bound for the PoS remains open.
Let $\F$ denote a family of latency functions and $\F$-TNGs and $\F$-RAGs denote, respectively, the family of TNGs and RAGs that use latency functions from this family.

\begin{theorem}
\label{thm:inefficiency}
Consider a family of latency functions $\F$. We have $PoS(\F\mbox{-TNGs}) \leq PoS(\F\mbox{-RAGs})$ and $PoA(\F\mbox{-TNGs}) \leq PoA(\F\mbox{-RAGs})$. In particular, the PoS and PoA for cost-sharing TNGs with $k$ players is at most $\log(k)$ and $k$, respectively, and for congestion TNGs with affine latency functions it is at most roughly $1.577$ and $\frac{5}{2}$
respectively.
\end{theorem}
\begin{proof}
We prove for PoS in cost-sharing games and the other proofs are similar. Consider a~TNG $\T$ and let $N^1, N^2,\ldots$ be a sequence of NEs whose cost tends to $c^* = \inf_{P \in \Gamma(\T)} cost(P)$. Let $O$ be a social optimum profile in $\T$, which exists due to Theorem~\ref{thm:SO-TNG-to-PTA}. Thus, $PoS(\T) = \lim_{j \to \infty} cost(N^j)/$ $cost(O)$. We show that each element in the sequence is bounded above by $PoS(\mbox{cost-sharing RAGs})$, which implies that $PoS(\T) \leq PoS(\mbox{cost-sharing RAGs})$, and hence
$PoS(\mbox{cost-sharing TNGs}) \leq PoS(\mbox{cost-sharing RAGs})$. 
For each $j \geq 1$, we construct 
below
an RAG $\R_j$ that has 
$PoS(\R_j) = cost(N^j)/cost(O)$, 
and since $\R_j$ is a cost-sharing RAG, we have $PoS(\R_j) \leq PoS(\mbox{cost-}$ $\mbox{sharing RAGs})$.

For $j \geq 1$, we construct a RAG $\R_j$ in which, for each $i \in [k]$, Player~$i$ has two strategies; one corresponding to her strategy in $N^j$ and one corresponding to her strategy in $O$. Formally, let $\Upsilon^j = \Upsilon_{N^j} \cup \Upsilon_O$ be the time periods of $N^j$ and $O$. We construct a RAG $\R_j = \zug{k, E^j, \set{\Sigma^j_i}_{i \in [k]}, \set{\ell_e}_{e \in E^j}}$, where $E^j = V \times \Upsilon^j$ and for 
each
$e = \zug{v, \gamma} \in E^j$, we 
define $\ell_e$ such that for each $l \in [k]$, we have that $\ell_e(l) = |\gamma| \cdot \ell_v(l)$.
and we define the players' strategies below.  Recall that, for a profile $P$, $i \in [k]$, and an interval $\gamma$, $\visits_P(i, \gamma)$ denotes the vertex at which Player~$i$ 
stays
during $\gamma$ in $P$. 
We extend the definition of $\visits_P$ to allow periods that occur after Player~$i$ has 
reached
her destination, and define that the function returns $u_i$. Moreover, we assume w.l.o.g. that $u_i$ is a vertex with no outgoing edges, thus the paths of the other players do not traverse $u_i$. Player~$i$'s two strategies are $n^j_i = \set{\visits_{N^j}(i, \gamma): \gamma \in \Upsilon_{N^j}}$ and $o^j_i = \set{\visits_O(i,\gamma): \gamma \in \Upsilon_O}$. Clearly, $\zug{o_1^j, \ldots, o^j_k}$ is the social optimum of $\R_j$. Also, $\zug{n^j_1,\ldots, n^j_k}$ is an NE and we assume it is the cheapest NE in $\R_j$. Otherwise, we can alter $N^j$ to match the best NE in $\R_j$ and only improve the sequence. Thus, we have $PoS(\R_j) = cost(N^j)/cost(O)$. Since $\R_j$ is a cost-sharing RAG, we have $PoS(\R_j) \leq PoS(\mbox{cost-sharing RAGs})$, and we are done.
\end{proof}

\section{Time Bounds}
\label{sec tb}
Recall that due to resets of clocks, the time by which a profile ends can be potentially unbounded.
It is interesting to know, given a TNG,  whether there are time bounds within which some interesting profiles like an NE and an SO are guaranteed to exist.
Earlier
we showed that every TNG is guaranteed to have an integral NE 
(Theorem~\ref{thm-all have bne}) 
and  an integral 
SO
(Theorem~\ref{thm:SO-TNG-to-PTA}). In this section we give bounds on the time 
by
which such profiles 
end.
That is, given a TNG $\T$, we find 
$t_{NE}(\T),T_{SO}(\T) \in \Q_{\ge 0}$
such that an integral NE $N$ and an integral SO $O$ exist in $\T$ in which the players reach their destinations by time $t_{NE}(\T)$ and $T_{SO}(\T)$ respectively. 

We start by showing a time bound on an optimal timed path in a PTA, and then proceed to TNGs.

\begin{lemma}
\label{lem:PTA-COR-time}
Consider a PTA $\P= \zug{C, V, E, \set{r_v}_{v \in V}}$, and let $\chi$ be the largest constant appearing in the guards on the edges of $\P$. Then, for every $s,u \in V$, there is an integral optimal timed path from $s$ to $u$ that ends by time $|V| \cdot (\chi+2)^{|C|}$.
\end{lemma}
\begin{proof}
Consider an optimal integral timed path $\eta$ in $\P$ that ends in the earliest time and includes no loop that is traversed instantaneously. Let $v_0,\ldots, v_n$ be the sequence of vertices that $\eta$ traverses, and, for 
$0 \leq i < n$, 
let $\kappa_i$ be the clock valuation before exiting the vertex $v_i$. Since $\eta$ is integral, $\kappa_i$ assigns integral values to clocks. Note that since the largest constant appearing in a guard in $\P$ is $\chi$, the guards in $\P$ cannot differentiate between clock values greater than $\chi$. We abstract away such values and define the {\em restriction} of a clock valuation $\kappa_i$ to be $\beta_i: C \rightarrow (\set{0} \cup [\chi] \cup \set{\top})$ by setting, for $x \in C$, the value $\beta_i(x) = \kappa_i(x)$, when $\kappa_i(x) \leq \chi$, and $\beta_i(x) = \top$, when $\kappa_i(x) > \chi$. Assume towards contradiction that $\eta$ ends after time $|V| \cdot (\chi+2)^{|C|}$. Then, there are $0 \leq i < j < n$ such that $\zug{v_i, \beta_i} = \zug{v_j, \beta_j}$. Let $\eta = \eta_1 \cdot \eta_2 \cdot \eta_3$ be a partition of $\eta$ such that $\eta_2$ is the sub-path between the $i$-th and $j$-th indices. Consider the path $\eta' = \eta'_1 \cdot \eta'_3$ that is obtained from $\eta$ by removing the sub-path $\eta_2$. First, note that $\eta'$ is a legal path. Indeed, the restrictions of the clock valuations in $\eta_1$ and $\eta_3$ match these in $\eta'_1$ and $\eta'_3$, that is, $\eta' = \eta_1 \cdot \eta_3$. Second, since we assume that traversing the loop $\eta_2$ is not instantaneous, we know that $\eta'$ ends before $\eta$. Moreover, since the rates in $\P$ are non-negative, we have $price(\eta') \leq price(\eta)$, and we reach a contradiction to the fact that $\eta$ is an optimal timed path that ends earliest. 
\end{proof}

\begin{theorem} \label{thm-SObound}
For a $k$-player TNG $\T$ with a set $V$ of vertices and a set $C$ of clocks, there exists an SO that ends 
by
time $\mathcal{O}(|V|^k \cdot \chi^{k|C|})$, where $\chi$ is the maximum constant appearing in $\T$. 
For every $k \geq 1$, there is a $k$-player 
(cost-sharing and congestion)
TNG ${\cal T}_k$ such that $\T_k$ has $\mathcal{O}(k)$ states, the boundaries in the guards in $\T_k$ are bounded by $\mathcal{O}(k \log k)$, and any SO in $\T_k$ requires time $2^{\Omega(k)}$.
\end{theorem}

\begin{proof}
We start with the upper bound. Consider a TNG $\T$
with a set $V$ of vertices and a set $C$ of clocks. By Theorem~\ref{thm:SO-TNG-to-PTA},
we can construct a PTA $\P$ 
with $|V|^ k$ vertices and $k|C|$ clocks 
such that a social optimum of $\T$ is an optimal timed path in $\P$. Applying Lemma~\ref{lem:PTA-COR-time},  we are done. 

We turn to the lower bounds. We show that for every $k \geq 1$, there is a $k$-player (cost-sharing and congestion) TNG ${\cal T}_k$ such that $\T_k$ has $\mathcal{O}(k)$ states, the boundaries in the guards in $\T_k$ are bounded by $\mathcal{O}(k \log k)$, and any SO in $\T_k$ requires time $2^{\Omega(k)}$.

\begin{figure}[ht]
\vspace{-2mm}
\includegraphics[width=0.5\textwidth]{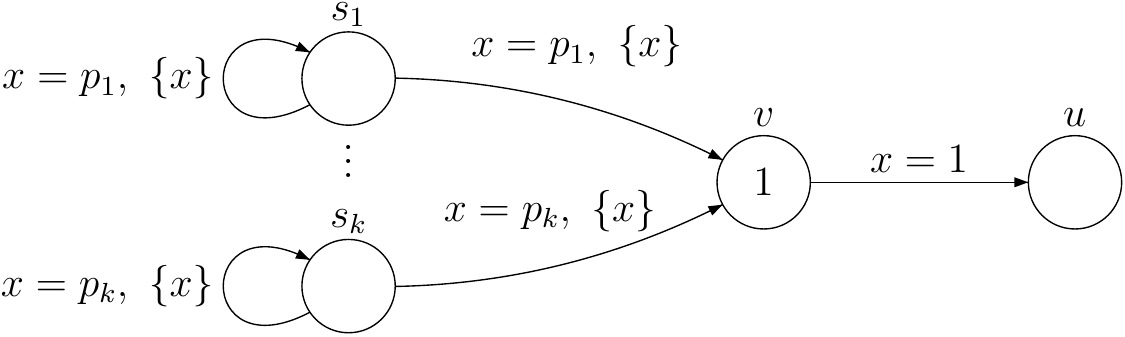}
\quad
\includegraphics[width=0.5\textwidth]{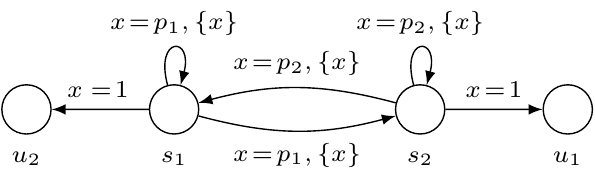}
\caption{\label{solong2} The time required for the SO is not polynomial.}
\end{figure}
\vspace{-.3cm}
Consider the $k$-player 
cost-sharing TNG  
appearing on the left of Figure~\ref{solong2}.
Let $p_1,\ldots,p_k$ be relatively prime (e.g., the the first $k$ prime numbers).
All the vertices in the TNG have cost $0$, except for $v$, which has some positive cost function. Each player $i$ has to spend  one time unit in $v$ in her path from $s_i$ to $u$. In an SO, all $k$ players spend this one time unit simultaneously, which forces them all to reach $v$ at time $\prod_{1 \leq i \leq k}  p_i$. Since the $i$-th prime number is $O(i \log i)$ and the product of the first $i$ prime numbers is $2^{\Omega(i)}$, we are done. 
We note that we could define the TNG also with no free vertices, 
that is vertics with $0$ cost,
by setting the cost in $v$ to be much higher than those in the source vertices.

For congestion games, the example
is more complicated. We start with the case of two players. Consider the 
congestion TNG
appearing on the right of
Figure~\ref{solong2}. Assume that $p_1$ and $p_2$ are relatively prime, $r_{s_1}(1)=r_{s_2}(1)=0$, and
$r_{s_1}(2)=r_{s_2}(2)=1$.
In the SO, the two players avoid each other in their paths from $s_i$ to $u_i$, and the way to do so is to wait $p_1 \cdot p_2$ time units before the edge from $s_i$ to $s_{3-i}$ is traversed.
Below
we generalize this example to $k$ players. Again, we could define the TNG with no free vertices.

We generalize the $2$-player congestion TNG appearing 
on the right
of Figure~\ref{solong2} to an arbitrary number of players. The extension to $3$ players appears in Figure~\ref{solong3} below.
As in the case of $2$ players, the cost function in the vertices $s_1, s_2$, and $s_3$ is  $0$ for load $1$ and strictly  positive for higher loads. In order to reach her target, Player $i$ has to traverse $2$ edges in the triangle before she can take the edge to $u_i$. In the SO that ends at the earliest possible time, the players perform these traversals together, so the game needs time $2 \cdot p_1 \cdot p_2 \cdot p_3$. 
 
\begin{figure}[ht]
\centering
\includegraphics[width=0.4\textwidth]{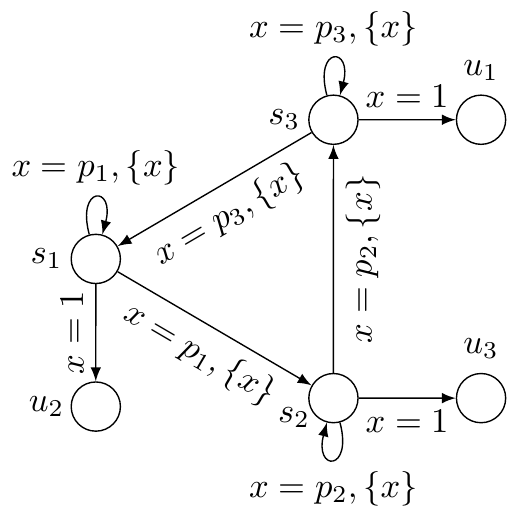}
\caption{\label{solong3} The time required for the SO is $2 \cdot p_1 \cdot p_2 \cdot p_3$.}
\end{figure}

In the extension to $k$ players, the TNG consists of a $k$-vertex polygon (to which the target vertices are connected), and the players have to traverse $k-1$ edges in it. Doing this simultaneously requires time $(k-1)  \cdot p_1 \cdot p_2 \cdots p_k$.
\end{proof}
\vspace{-.3cm}

We proceed to derive a time bound for the existence of an NE. For a TNG $\T$, let $L_\T \in \Nat$ be the smallest number such that multiplying the latency functions by $L_\T$ results in a normalized TNG. 
Recall the $SO(\T)$ is the cost of a social optimum in $\T$.

\begin{theorem}
\label{thm:NE-time-bound}
Consider a TNG $\T$ with $k$ players, played on a timed network $\zug{V, E, C}$, and let $\chi$ be the maximum constant appearing in a guard. Then, there is an NE in $\T$ that ends by time $\mathcal{O}(\varphi \cdot |V| \cdot \chi^{|C|} + |V|^k \cdot \chi^{k|C|})$, where $\varphi= L_\T \cdot SO(\T)$ for congestion TNGs and $\varphi = L_\T \cdot \log(k) \cdot SO(\T)$ for cost-sharing TNGs.
\stam{
Consider a TNG $\T = \zug{k, C, V, E, \{\ell_v\}_{v \in V}, \zug{s_i, u_i}_{i \in [k]}}$, where $\chi$ is the maximum constant appearing in the guards on the edges of $\T$ and SOPT is the cost of a social optimum profile. Then, there is an NE in $\T$ that ends by time $\mathcal{O}(|V|^{2^{\varphi-1}k+2^{\varphi-2}} + \chi^{(2^ \varphi - 1)k|C|})$,
where for congestion TNGs, we have $\varphi= L_\T \cdot (|C|k)^{SOPT}$ and for cost-sharing TNGs, we have $\varphi = L_\T \cdot (|C|k)^{\log(k) \cdot SOPT}$.}
\end{theorem}
\begin{proof}
Recall the proof of Theorem~\ref{thm-all have bne} that shows that every TNG has an integral NE: we choose an initial integral profile $P$ and perform integral best-response moves until an NE is reached. The number of iterations is bounded by the potential $\Psi(P)$ of $P$. We start the best-response sequence from a social-optimum profile $O$ that ends earliest. By Theorem~\ref{thm-SObound}, there is such a profile that ends by time $\mathcal{O}(|V|^k \cdot \chi^{k|C|})$. Let $\varphi = L_\T \cdot SO(\T)$ in the case of congestion TNGs and $\varphi = L_\T \cdot (\ln(k)+1) \cdot SO(\T)$ in the case of cost-sharing TNGs. It is not hard to show that $\Psi(O) \leq \varphi$.

Next, we bound the time that is added in a best-response step. We recall the construction in Theorem~\ref{thm:BR-TNG-to-PTA} of the PTA $\P$ for finding a best-response move. Consider a TNG $\T$ and a profile of strategies $P$, where, w.l.o.g., we look for a best-response for Player~$k$. Suppose the strategies of Players $1,\ldots, k-1$ take transitions at times $\tau_1,\ldots, \tau_n$. We construct a PTA $\P$ with $n+1$ copies of $\T$. For $1 \leq i \leq n+1$, an optimal path in $\P$ starts in the first copy and moves from copy~$i$ to copy~$(i+1)$ at time $\tau_i$. We use the additional ``global'' clock to enforce these transitions. A key observation is that in the last copy, this additional clock is never used. Thus, the largest constant in a guard in the last copy coincides with $\chi$, the largest constant appearing in $\T$. Let $\eta$ be an optimal path in $\P$ and $\pi_k$ the corresponding strategy for Player~$k$. We distinguish between two cases. If $\eta$ does not enter the last copy of $\P$, then it ends before time $\tau_n$, namely the latest time at which a player reaches her destination. Then, the profile $P[k \gets \pi_k]$ ends no later than $P$. In the second case, the path $\eta$ ends in the last copy of $\P$. We view the last copy of $\P$ as a PTA. By Lemma~\ref{lem:PTA-COR-time}, the time at which $\eta$ ends is within $|V| \cdot (\chi+2)^{|C|}$ since its entrance into the copy, which is $\tau_n$. Then, $P[i \gets \pi_k]$ ends at most $|V| \cdot (\chi+2)^{|C|}$ time units after $P$. To conclude, the best-response sequence terminates in an NE that ends by time $\mathcal{O}(\varphi \cdot |V| \cdot (\chi+2)^{|C|} + |V|^k \cdot \chi^{k|C|})$.
\end{proof}

\section{Discussion and Future Work}
\label{disc}
The model of TNGs studied in this paper extends the model of GTNGs introduced in \cite{AGK17} by adding clocks. From a practical point of view, the addition of clocks makes TNGs significantly more expressive than GTNGs and enables them to model the behavior of many systems that cannot be modeled using GTNGs. From a theoretical point of view, the analysis of TNGs poses different and difficult technical challenges. In the case of GTNGs, a main tool for obtaining positive results is a reduction between GTNGs and NGs. Here, in order to obtain positive results we need to combine techniques from NGs and PTAs.

\stam{
We left several open problems. In Theorem \ref{thm-all have bne}, we describe a method for finding an integral NE through a sequence of BR moves. We leave open the complexity of finding an NE in TNGs. The complexity analysis of search problems typically refers to the complexity of an improvement step in a local search. For example, PLS \cite{FPT04}, which lies ``close'' to P, consists  of problems in which an improvement step can be done in polynomial time. For TNGs, an improving step, namely finding a best-response, is PSPACE-hard, and we conjecture that the problem is PSPACE-complete.
Further we show that the BR and SO problems for one-clock TNGs are in PSPACE and is NP-hard, leaving open the tight complexity. 
}

We left several open problems. In Theorem \ref{thm-all have bne}, we describe a method for finding an integral NE through a sequence of BR moves. We leave open the complexity of finding an NE in TNGs. For the upper bound, we conjecture that there is a PSPACE algorithm for the problem. For the lower bound, we would need to find an appropriate complexity class of search problems and show hardness for that class. For example, PLS \cite{FPT04}, which lies ``close'' to P, and includes the problem of finding an NE in NGs, consists  of search problems in which a local search, e.g., a BR sequence, terminates. Unlike NGs, where a BR can be found in polynomial time, in TNGs, the problem is PSPACE-complete. To the best of our knowledge, complexity classes for search problems that are higher than PLS were not studied.
Further we show that the BR and SO problems for one-clock TNGs is in PSPACE and is NP-hard, leaving open the tight complexity.

This work belongs to a line of works that transfer concepts and ideas between the areas of formal verification and algorithmic game theory: logics for specifying multi-agent systems \cite{AHK02,CHP07}, studies of equilibria in games related to synthesis and repair problems \cite{CHJ06,Cha06,FKL10,AAK15b}, and 
of non-zero-sum games in formal verification \cite{CMJ04,BBPG12}. This line of work also includes efficient reasoning about NGs with huge networks \cite{KT17,AGK17b}, an extension of NGs to objectives that are richer than reachability \cite{AKT16}, and NGs in which the players select their paths dynamically \cite{AHK16}. For future work, we plan to apply the real-time behavior of TNGs to these last two concepts; namely, TNGs in which the players' objectives are given as a specification that is more general than simple reachability or TNGs in which the players reveal their choice of timed path in steps, bringing TNGs closer to the timed games of \cite{AM99,AFHMS03}. 

\bibliography{ok}
 
\end{document}